\documentclass[12pt,reqno]{article}

\usepackage[usenames]{color}
\usepackage{amssymb}
\usepackage{amsmath}
\usepackage{amsthm}
\usepackage{amsfonts}
\usepackage{amscd}
\usepackage{graphicx}
\usepackage[shortlabels]{enumitem}

\usepackage[colorlinks=true,
linkcolor=webgreen,
filecolor=webbrown,
citecolor=webgreen]{hyperref}

\definecolor{webgreen}{rgb}{0,.5,0}
\definecolor{webbrown}{rgb}{.6,0,0}

\usepackage{color}
\usepackage{fullpage}
\usepackage{float}

\usepackage{graphics}
\usepackage{latexsym}
\usepackage{epsf}
\usepackage{breakurl}

\def\Enn{\mathbb{N}}

\def\suchthat{\, : \, }

\begin{document}

\theoremstyle{plain}
\newtheorem{theorem}{Theorem}
\newtheorem{corollary}[theorem]{Corollary}
\newtheorem{lemma}[theorem]{Lemma}
\newtheorem{proposition}[theorem]{Proposition}

\theoremstyle{definition}
\newtheorem{definition}[theorem]{Definition}
\newtheorem{example}[theorem]{Example}
\newtheorem{conjecture}[theorem]{Conjecture}

\theoremstyle{remark}
\newtheorem{remark}[theorem]{Remark}

\title{Abelian Complexity and Synchronization}

\author{Jeffrey Shallit\\
School of Computer Science \\
University of Waterloo \\
Waterloo, ON  N2L 3G1 \\
Canada \\
\href{mailto:shallit@uwaterloo.ca}{\tt shallit@uwaterloo.ca} }

\maketitle

\begin{abstract}
We present a general method for computing the abelian 
complexity $\rho^{\rm ab}_{\bf s} (n)$
of an automatic sequence $\bf s$ in the case
where (a) $\rho^{\rm ab}_{\bf s} (n)$ is bounded by a constant and (b)
the Parikh vectors of the length-$n$ prefixes of $\bf s$ form
a synchronized sequence.

We illustrate the idea in detail,
using the free software {\tt Walnut} to
compute the abelian complexity of 
the Tribonacci word ${\bf TR} = 0102010\cdots$, 
the fixed point of the morphism $0 \rightarrow 01$, $1 \rightarrow 02$,
$2 \rightarrow 0$.
Previously, Richomme, Saari, and Zamboni showed that the
abelian complexity of this word lies in $\{ 3,4,5,6,7 \}$, and
Turek gave a Tribonacci automaton computing it.  We are able to
``automatically''
rederive these results, and more, using the method presented here.
\end{abstract}

\section{Introduction}

Let $\Sigma = \{ a_1, a_2, \ldots, a_k \}$ be a finite ordered alphabet,
with an ordering on the letters given by
$a_1 < a_2 < \cdots < a_k$.   Let $w \in \Sigma^*$.
We define $|w|_{a_i}$ to be the number of occurrences of $a_i$ in
$w$.   Thus, for example, $|{\tt banana}|_{\tt a} = 3$.

The {\it Parikh vector}
$\psi(w)$ for $w \in \Sigma^*$ is defined to be 
$(|w|_{a_1}, \ldots, |w|_{a_k})$, the $k$-tuple that
counts the number of occurrences of each letter in $w$.
Thus, for example, if ${\tt v} < {\tt l} < {\tt s} < {\tt e} $, then
$\psi({\tt sleeveless}) = (1,2,3,4)$.

Let ${\bf s} \in \Sigma^\omega$ be an infinite sequence over $\Sigma$.
The {\it abelian complexity function} $\rho^{\rm ab}_{\bf s} (n)$ is defined
to be the number of distinct Parikh vectors of length-$n$ factors
of $\bf s$.   This concept was introduced by
Richomme, Saari, and Zamboni \cite{Richomme&Saari&Zamboni:2011}
and has been studied extensively since then.

We call a numeration system for $\Enn$ {\it regular} if
\begin{itemize}
\item[(a)]  Elements of $\Enn$ have unique representations
(up to leading zeros);

\item[(b)]  The set of all valid representations for $\Enn$ is recognizable by
a finite automaton;

\item[(c)]  The exists another  finite automaton recognizing the relation
$\{ (x,y,z) \suchthat x = y + z \}$.
By this we mean that 
there is an automaton accepting, in parallel, the representations
of $x, y, z$ satisfying $x = y+z$.    Here any shorter representation
is padded with leading zeros, if necessary, so that the representations of
all three numbers have the same length, and can be fed into the
automaton digit-by-digit, in parallel.
\end{itemize}

Examples of regular numeration systems include base $k$, for integers
$k \geq 2$ \cite{Allouche&Shallit:2003};
Fibonacci numeration \cite{Frougny:1986}; Tribonacci numeration
\cite{Mousavi&Shallit:2015}; and Ostrowski numeration systems
\cite{Baranwal:2020}.

In this paper we will be concerned with {\it automatic sequences}.
These are sequences $\bf s$ that can be computed by a DFAO (deterministic
finite automaton with output) $M$ in the following sense:  we express
$n$ in some regular numeration system, and feed $M$
with this representation, most significant digit first.
The output associated with the last
state reached is then ${\bf s}[n]$, the $n$'th term of the sequence
$s$.

A variation on automatic sequences is the {\it synchronized sequence},
introduced by Carpi and Maggi \cite{Carpi&Maggi:2001}.  A sequence
$(s(n))_{n \geq 0}$ taking values in $\Enn$ (or $\Enn^k$) is synchronized if
there is an automaton accepting, in parallel, the representations
of $n$ and $s(n)$, with shorter representations padded with leading
zeroes, as above.

Abelian complexity, defined above,
is a variation on (ordinary) subword complexity
$\rho_{\bf s}(n)$, which counts the number of length-$n$ factors
of $\bf s$.   It is known that for automatic sequences, the first
difference of the subword complexity function is ``automatically
computable'', in the sense that there is an algorithm that,
given the DFAO $M$, will construct another automaton $M'$ that
computes $\rho_{\bf s}(n+1) - \rho_{\bf s}(n)$ 
\cite{Goc&Schaeffer&Shallit:2013}.
In contrast,
for abelian complexity there is (in general) no such automaton,
as proved in \cite[Cor.~5.6]{Schaeffer:2013}.

In this note I will show that there is an algorithm
that, given an automatic sequence $\bf s$ satisfying
certain conditions, will
construct a DFAO computing $\rho_{\bf s}^{\rm ab} (n)$.
I will illustrate the ideas in detail for the
Tribonacci word, fixed point of the morphism
$0 \rightarrow 01$, $1 \rightarrow 02$, $2 \rightarrow 0$.

\section{The result and its proof}
\label{steps}

\begin{theorem}
Let $\bf s$ be a sequence that is automatic in some regular numeration system.
Suppose that
\begin{itemize}
\item[(a)] the abelian complexity
$\rho^{\rm ab}_{\bf s} (n)$ is bounded above by a constant, and 
\item[(b)]
the Parikh vectors of
length-$n$ prefixes of $\bf s$ form a synchronized sequence.
\end{itemize}
Then $\rho^{\rm ab}_{\bf s} (n)$ is an automatic sequence and the
DFAO computing it is effectively computable.
\end{theorem}

\begin{proof}
Here is a summary of how the DFAO computing $\rho_{\bf s}^{\rm ab} (n)$
can be automatically constructed.  We use the fact that a first-order
logical formula $\varphi$ with addition, and using indexing into an automatic
sequence, is itself automatic (that is, there is an automaton accepting
the values of the free variables that make $\varphi$ true)
\cite{Charlier&Rampersad&Shallit:2011}.

\begin{enumerate}
\item Given a synchronized automaton computing the Parikh vector
of length-$n$ prefixes of $\bf s$, we can find
(by subtracting) a synchronized automaton computing the Parikh vector for
arbitrary length-$n$ factors ${\bf s}[i..i+n-1]$.   This is
first-order expressible.   The resulting automaton accepts
(in parallel) triples of the form $(i,n, \psi({\bf s}[i..i+n-1]))$.

\item The quantity $\rho^{\rm ab}_{\bf s} (n)$ is bounded 
iff there is a constant $C$ (not depending on $n$) such that the
cardinality of each set
$$A_n := \{ \psi({\bf s}[i..i+n-1]) - \psi({\bf s}[0..n-1]) 
\suchthat i \geq 0 \}$$
is bounded above by $C$.

\item In this case, the range of each coordinate
of $A_n$ is therefore finite,
and can be computed algorithmically.

\item Once we know the ranges of each coordinate, there are 
only finitely many possibilities for 
$\psi({\bf s}[i..i+n-1]) - \psi({\bf s}[0..n-1])$.
We can then compute the set of all possibilities $S$.

\item Once we have $S$, we can test each of the finitely many
subsets to see if it occurs for some $n$, and obtain
an automaton recognizing those $n$ for which it does.

\item All the different automata can then be combined into a single
DFAO computing $\rho_{\bf s}^{\rm ab} (n)$,
using the direct product construction
discussed in \cite[Lemma 2.4]{Shallit&Zarifi:2019}.
\end{enumerate}
\end{proof}

With the exception of the last step, all these steps can be achieved
using the free software {\tt Walnut}.  We implemented
the last step in {\tt Dyalog APL}, although it should be a feature
of a future version of {\tt Walnut}.  

We now illustrate
the ideas in detail on a particular infinite word, the Tribonacci word.
In doing so, we recover almost all of the results from two papers:
\cite{Richomme&Saari&Zamboni:2010} and \cite{Turek:2015}. In particular,
we avoid the long, {\it ad hoc}, and case-based reasoning of the first paper;
it is replaced by calculations done ``automatically'' by {\tt Walnut}.
The second paper, by Turek, is much closer to our approach in spirit
and details, but still has some {\it ad hoc} aspects.
In an earlier paper on Tribonacci-automatic sequences
\cite{Mousavi&Shallit:2015}, we stated that we could not yet obtain
the abelian complexity of the Tribonacci word, a deficiency we remedy here.

\section{The Tribonacci word and Tribonacci representations}

The infinite Tribonacci word ${\bf TR} = {\bf TR}[0..\infty) =
0102010\cdots$ is the fixed point of the map
$0 \rightarrow 01$, $1 \rightarrow 02$, $2 \rightarrow 0$.
It was studied, for example, in 
\cite{Chekhova&Hubert&Messaoudi:2001,Barcucci&Belanger&Brlek:2004,Rosema&Tijdeman:2005}.

The word $\bf TR$ is Tribonacci-automatic.  This means there is 
a DFAO that takes the Tribonacci representation of $n$
as input, and computes ${\bf TR}[n]$, the $n$'th term
of the Tribonacci sequence.   (Indexing starts with $0$.)  Here the Tribonacci
representation of a natural number $n$ is a binary word
$w = e_1 e_2 \cdots e_r$ such that
$n = \sum_{1 \leq i \leq r} T_{r+2-i}$,
where $(T_i)_{i \geq 0}$ is the Tribonacci sequence,
defined by $T_0 = 0$, $T_1 = 1$, $T_2 = 1$, and
$T_n = T_{n-1} + T_{n-2} + T_{n-3}$ for $n \geq 3$.
The Tribonacci representation for $n$ is unique,
provided the word
$w$ begins with $1$ (for $n \geq 1$) and contains no block of the form $111$;
see \cite{Carlitz&Scoville&Hoggatt:1972}.  We write this
unique representation as $(n)_T$.  The inverse map
$[w]_T$ maps a binary word to the sum
$\sum_{1 \leq i \leq r} T_{r+2-i}$.

More specifically, we have that 
$$ {\bf TR}[n] =
\begin{cases}
0, & \text{if $(n)_T$ ends in $0$}; \\
1, & \text{if $(n)_T$ ends in $01$}; \\
2, & \text{if $(n)_T$ ends in $11$}.
\end{cases} $$

To see that the Parikh vector of length-$n$ prefixes of
$\bf TR$ is synchronized, it suffices to show that
each of the three functions $n \rightarrow |{\bf TR}[0..n-1]|_a$
is synchronized, for $a\in \{ 0,1,2 \}$.  This is easily
seen using the following relations.  Write
$(n)_T = e_1 e_2 \cdots e_r$.  Then
\begin{align}
|{\bf TR}[0..n-1]|_0 &= [e_1\cdots e_{r-1}]_T + e_r \nonumber \\
|{\bf TR}[0..n-1]|_1 &= [e_1\cdots e_{r-2}]_T + e_{r-1} \label{sync} \\
|{\bf TR}[0..n-1]|_2 &= [e_1\cdots e_{r-3}]_T + e_{r-2} \nonumber
\end{align}
See, for example, \cite[Thm.~20]{Mousavi&Shallit:2015} and
\cite[Thm.~13]{Dekking&Shallit&Sloane:2020}.

\section{{\tt Walnut} implementation of the algorithm for {\bf TR}}

In this section we show how to implement our ideas using the
word {\bf TR} as an example.   The code is given in {\tt Walnut},
a theorem-prover for first-order logical formulas on automatic
sequences \cite{Mousavi:2016}.

Let us start by implementing the formulas \eqref{sync} given above.

In order to do this, we need to be able to tell whether a number
is a right-shift of another (in Tribonacci representation).
For example, $6 = [110]_T$ is the right shift of $12 = [1101]_T$.
It is easy to create a 2-D Tribonacci automaton for this, and it is
illustrated below in Figure~\ref{one}.  It takes the Tribonacci
representations of $m$ and $n$ in parallel, and accepts
if $(n)_T$ is the right shift of $(m)_T$.  
We call the result
{\tt \$rst} and it is stored in the {\tt Result} directory of
{\tt Walnut}.
\begin{figure}[H]
\begin{center}
\vskip -.3in
\includegraphics[width=4.5in]{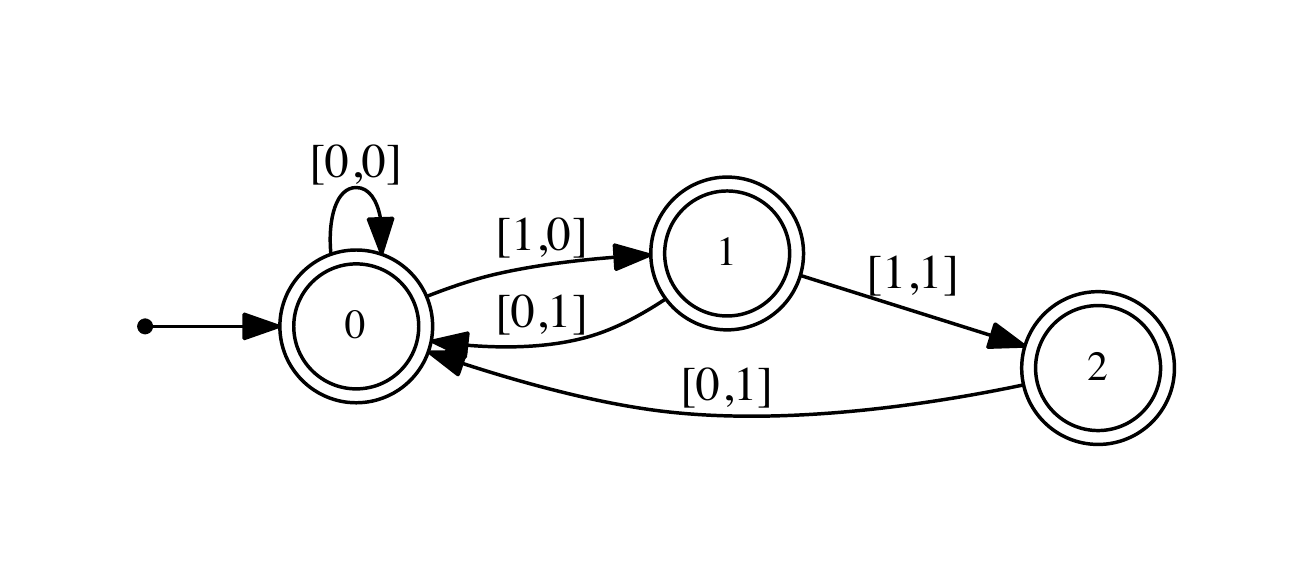}
\end{center}
\vskip -.5in
\caption{Tribonacci shift is synchronized.}
\label{one}
\end{figure}
We also need to be able to compute the last bit of $n$ in Tribonacci
representation.   
Again, this is very easy, and an automaton 
computing it is illustrated below in Figure~\ref{two}.  We
call the result {\tt TRL}, and the appropriate file is stored
under the name {\tt TRL.txt} in the {\tt Word Automata Library} of
{\tt Walnut}.
\begin{figure}[H]
\begin{center}
\vskip -.3in
\includegraphics[width=4.5in]{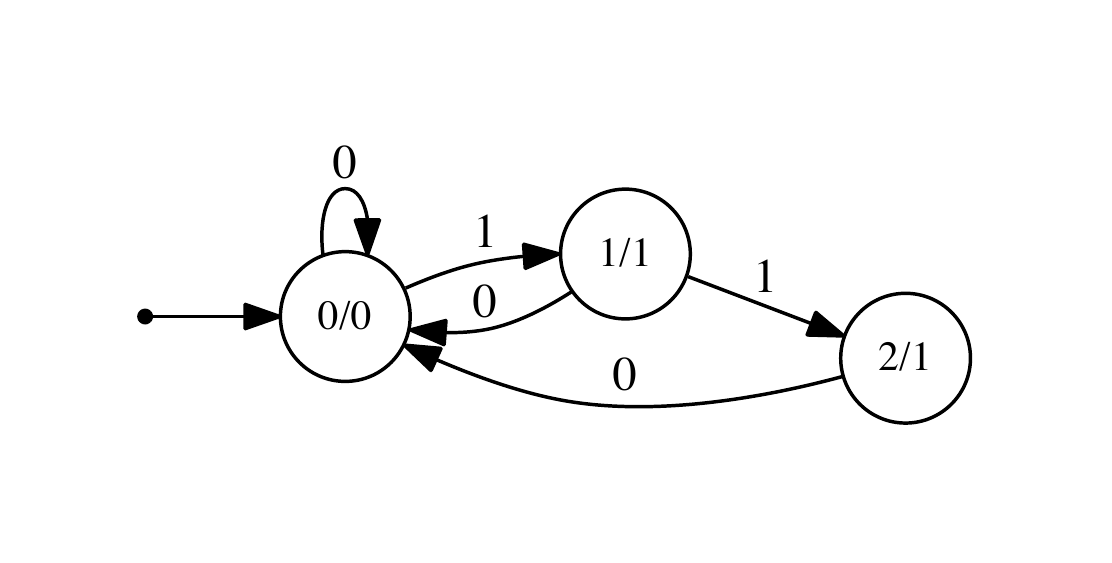}
\end{center}
\vskip -.5in
\caption{DFAO for the last bit.}
\label{two}
\end{figure}

Now, using this, we can find DFAO's computing the maps
$n \rightarrow |{\bf TR}[0..n-1]|_a$ 
for $a \in \{ 0, 1, 2\}$.  Here is the {\tt Walnut} code:
\begin{verbatim}
def tribsync0 "?msd_trib Ea Eb (s=a+b) & ((TRL[n]=@0)=>b=0) &
     ((TRL[n]=@1)=>b=1) & $rst(n,a)":
def tribsync1 "?msd_trib Ea Eb Ec (s=b+c) & ((TRL[a]=@0)=>c=0) &
     ((TRL[a]=@1)=>c=1) & $rst(n,a) & $rst(a,b)":
def tribsync2 "?msd_trib Ea Eb Ec Ed (s=c+d) & ((TRL[b]=@0)=>d=0) &
     ((TRL[b]=@1)=>d=1) & $rst(n,a) & $rst(a,b) & $rst(b,c)":
\end{verbatim}

This gives three synchronized automata computing
$|{\bf TR}[0..n-1]|_a$ for $a \in \{ 0,1,2 \}$, as follows:
\begin{figure}[H]
\begin{center}
\vskip -.3in
\includegraphics[height=3in]{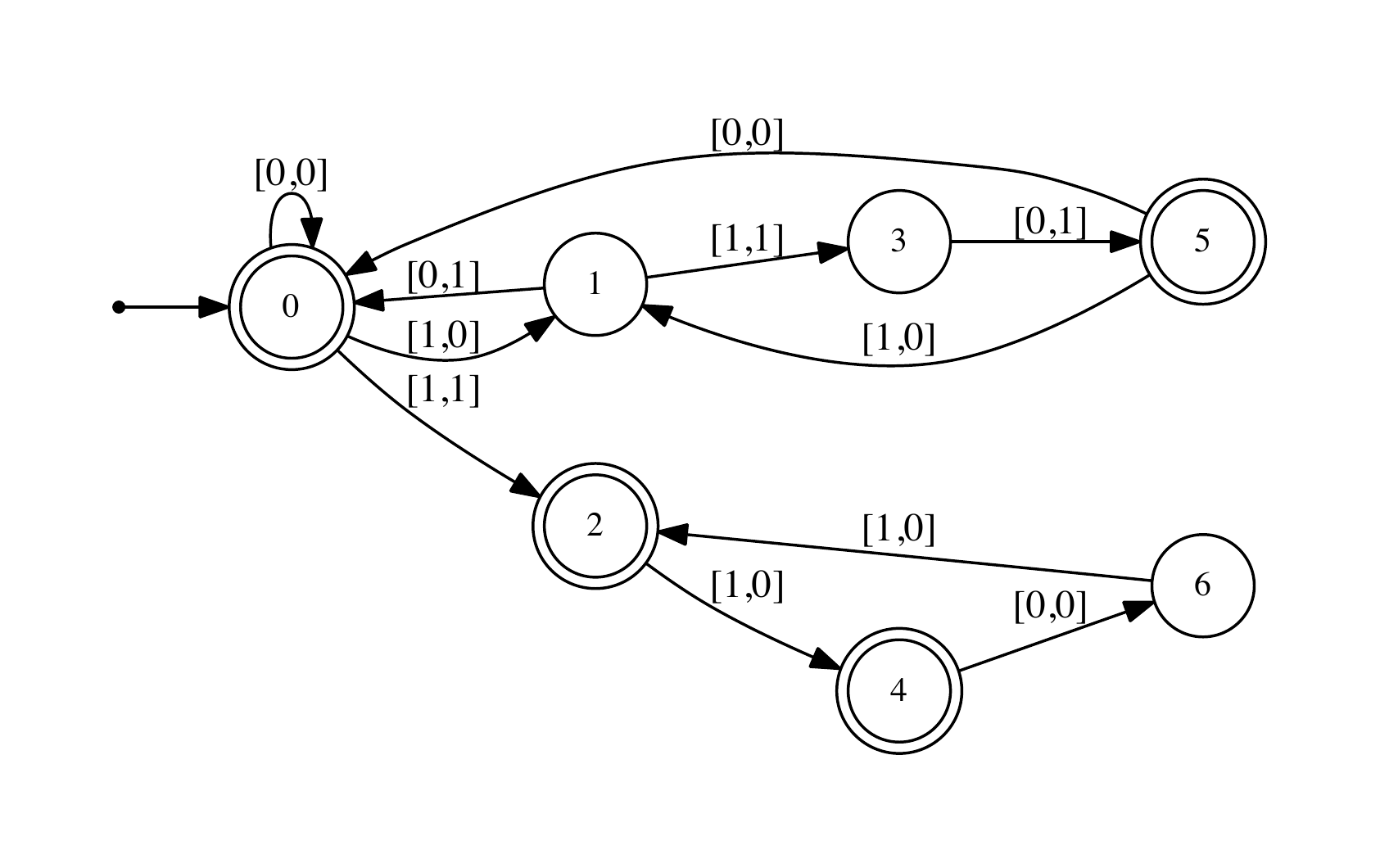}
\end{center}
\vskip -.5in
\caption{Synchronized automaton for $|{\bf TR}[0..n-1]|_0$.}
\end{figure}
\begin{figure}[H]
\begin{center}
\vskip -.3in
\includegraphics[height=3.5in]{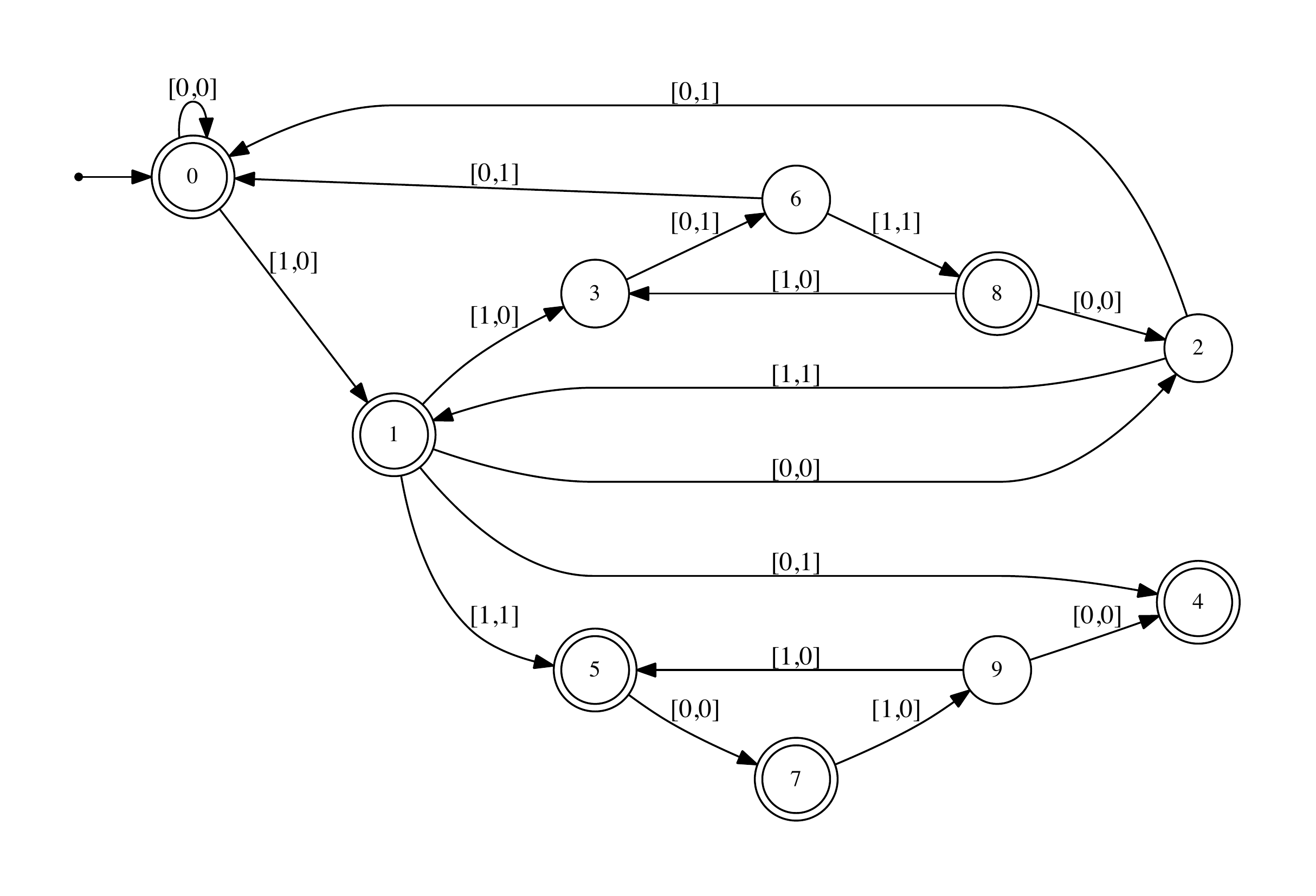}
\end{center}
\vskip -.5in
\caption{Synchronized automaton for $|{\bf TR}[0..n-1]|_1$.}
\end{figure}
\begin{figure}[H]
\begin{center}
\vskip -.3in
\includegraphics[height=4in]{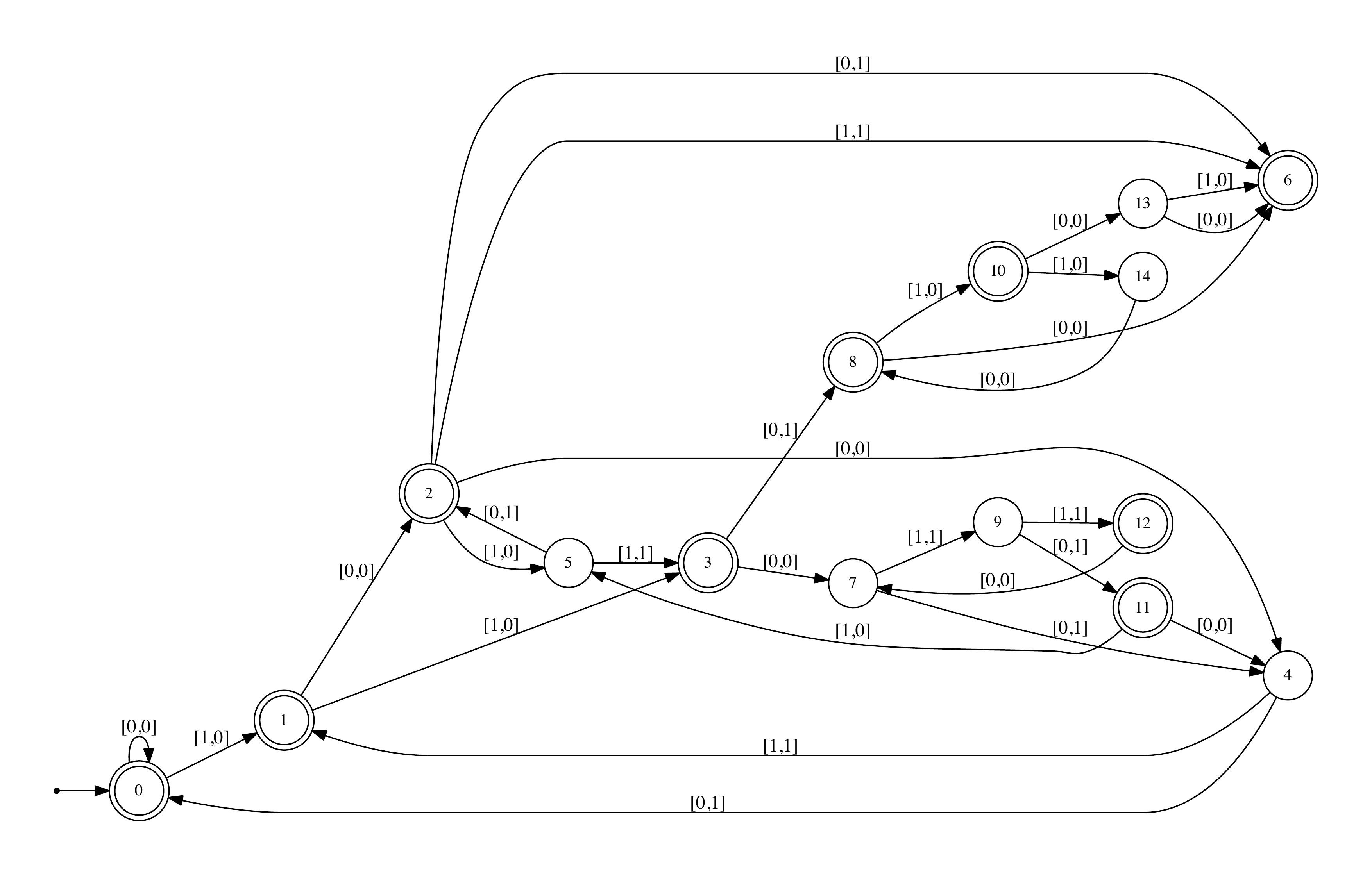}
\end{center}
\vskip -.5in
\caption{Synchronized automaton for $|{\bf TR}[0..n-1]|_2$.}
\end{figure}

(We could, if we wished, combine this into one synchronized automaton
computing all three
elements of the Parikh vector of ${\bf TR}[0..n-1]$, but the resulting
automaton has $31$ states and is a little awkward to display.)

Next we compute synchronized
Tribonacci automata computing $|{\bf TR}[i..i+n-1]|_a$ 
for $a \in \{ 0,1,2 \}$.  
We can do this with the following first-order formulas
computing 
$$|{\bf TR}[0..i+n-1]|_a - |{\bf TR}[0..i-1]|_a,$$
namely
{\small
\begin{verbatim}
def tribfac0 "?msd_trib Aq Ar ($tribsync0(i+n,q) & $tribsync0(i,r)) => (q=r+s)":
def tribfac1 "?msd_trib Aq Ar ($tribsync1(i+n,q) & $tribsync1(i,r)) => (q=r+s)":
def tribfac2 "?msd_trib Aq Ar ($tribsync2(i+n,q) & $tribsync2(i,r)) => (q=r+s)":
\end{verbatim}
}
The resulting automata have $239$, $283$, and $406$ states, respectively --- much too large to display here.

We now move on to computing the abelian complexity of $\bf TR$.  
We want to compute the number of distinct triples 
$\psi({\bf TR}[i..i+n-1])$ for $i \geq 0$.   This is the same
as the number of distinct triples
$$f(i,n) := \psi({\bf TR}[i..i+n-1]) - \psi({\bf TR}[0..n-1])$$
for $i \geq 0$.   (These were called \textit{relative Parikh vectors}
by Turek \cite{Turek:2015}.)

In order to compute this, we first need to know the range of
each coordinate.   Since this range could include negative integers,
and {\tt Walnut} is currently restricted to $\Enn$, this is slightly
awkward, but can be done as follows:
{\small
\begin{verbatim}
def posrange0 "?msd_trib E i,n,s,t $tribfac0(i,n,s) & $tribfac0(0,n,t) & u+t=s":
def negrange0 "?msd_trib E i,n,s,t $tribfac0(i,n,s) & $tribfac0(0,n,t) & u+s=t":
def posrange1 "?msd_trib E i,n,s,t $tribfac1(i,n,s) & $tribfac1(0,n,t) & u+t=s":
def negrange1 "?msd_trib E i,n,s,t $tribfac1(i,n,s) & $tribfac1(0,n,t) & u+s=t":
def posrange2 "?msd_trib E i,n,s,t $tribfac2(i,n,s) & $tribfac2(0,n,t) & u+t=s":
def negrange2 "?msd_trib E i,n,s,t $tribfac2(i,n,s) & $tribfac2(0,n,t) & u+s=t":
\end{verbatim}
}
By inspecting the resulting six automata, we see that
\begin{equation}
f(i,n) \in \{ -1, 0, 1 \} \times \{ -1,0, 1, 2\} \times
	\{ -1,0,1, 2 \},
\label{range}
\end{equation}
which proves that $\bf TR$ has bounded abelian complexity and
also proves that $\bf TR$ is 2-balanced
\cite[Thm.~1.3]{Richomme&Saari&Zamboni:2010}.

We now proceed to determine which specific triples can appear in the range 
of $f(i,n)$.  Since we know the range from \eqref{range},
we can do this as follows:
\begin{verbatim}
def validtriples "?msd_trib Ei,n,a,b,c,d,e,f $tribfac0(i,n,a) &
   $tribfac0(0,n,b) & s+b=a+1 & $tribfac1(i,n,c) & $tribfac1(0,n,d)
   & t+d=c+1 & $tribfac2(i,n,e) & $tribfac2(0,n,f) & u+f=e+1":
\end{verbatim}
This returns the range of $f(i,n)+(1,1,1)$, which from \eqref{range} is guaranteed
to be in $\Enn^3$.  The resulting automaton is as follows:
\begin{figure}[H]
\begin{center}
\vskip -.3in
\includegraphics[width=4.5in]{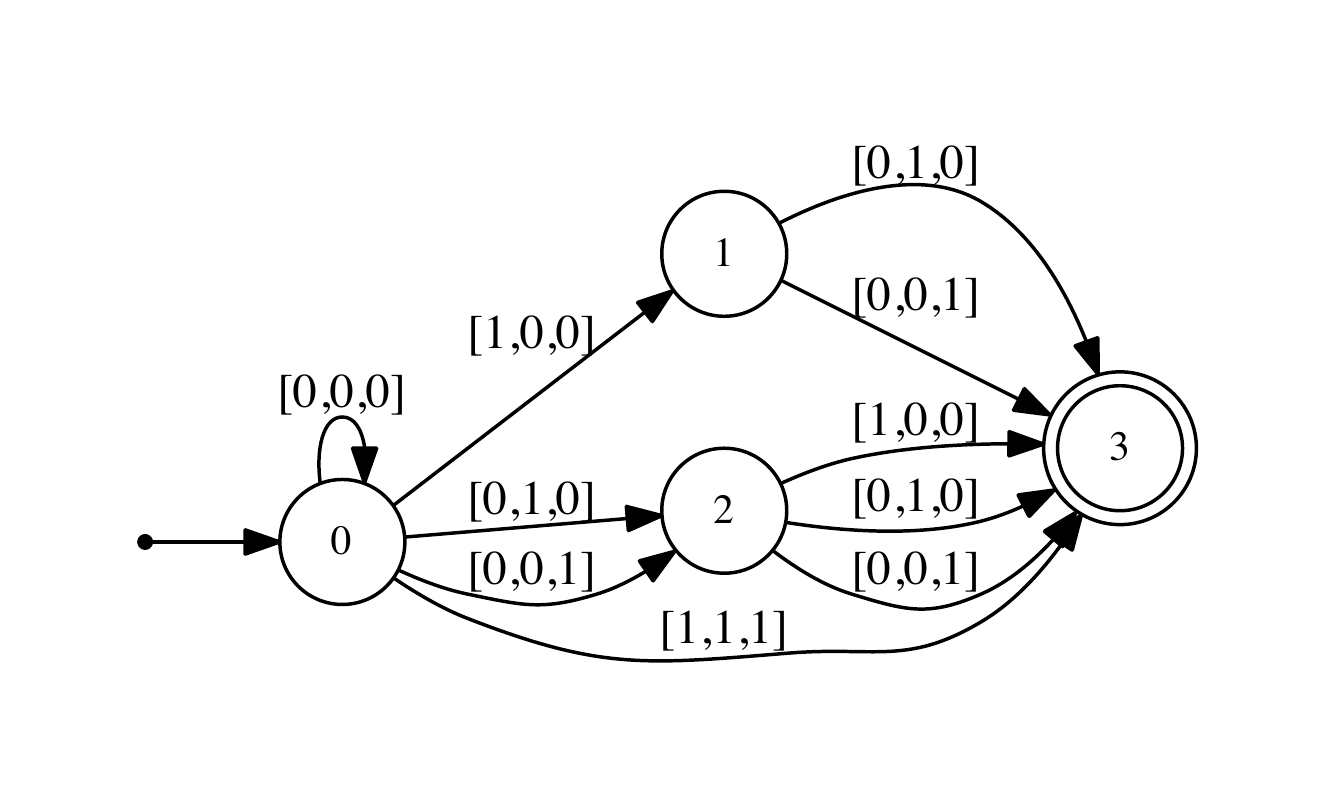}
\end{center}
\vskip -.5in
\caption{Range of $f(i,n) + (1,1,1)$.}
\label{validt}
\end{figure}
As you can see by inspecting the automaton in Figure~\ref{validt}, the
automaton accepts the Tribonacci representation of exactly nine
triples, namely
$$\{ (1,1,1), (2,1,0), (2,0,1), (1,2,0), (0,3,0), (0,2,1), (1,0,2), (0,1,2), (0,0,3) \}.$$
Thus we have proven that the range of $f(i,n)$ is $A$, where
\begin{align*}
A &:=  \{ (0,0,0), (1,0,-1), (1,-1,0), (0,1,-1), (-1,2,-1), \\
& \quad\quad (-1,1,0), (0,-1,1), (-1,0,1), (-1,-1,2)  \}.
\end{align*}

Define $A_n = \{ f(i,n) \suchthat i \geq 0 \}$.
It now remains to see which of the $2^9$ possible subsets of $A$ can
be an $A_n$ for some $n$.
(Actually, since
$(0,0,0)$ occurs for each $n$, a priori there are only $2^8$ possible subsets.)
We could do this by blindly following the recipe in
Section~\ref{steps}, but since most subsets are not reachable, we
can use a faster way.

First let's make {\tt Walnut} formulas for 
the assertion that $f(i,n)$ equals each of the nine tuples listed
\begin{verbatim}
def t000 "?msd_trib Ea,b,c,d,e,f $tribfac0(i,n,a) &
   $tribfac0(0,n,b) & a=b & $tribfac1(i,n,c) & $tribfac1(0,n,d)
   & c=d & $tribfac2(i,n,e) & $tribfac2(0,n,f) & e=f":
def t10m1 "?msd_trib Ea,b,c,d,e,f $tribfac0(i,n,a) &
   $tribfac0(0,n,b) & a=b+1 & $tribfac1(i,n,c) & $tribfac1(0,n,d)
   & c=d & $tribfac2(i,n,e) & $tribfac2(0,n,f) & e+1=f":
def t1m10 "?msd_trib Ea,b,c,d,e,f $tribfac0(i,n,a) &
   $tribfac0(0,n,b) & a=b+1 & $tribfac1(i,n,c) & $tribfac1(0,n,d)
   & c+1=d & $tribfac2(i,n,e) & $tribfac2(0,n,f) & e=f":
def t01m1 "?msd_trib Ea,b,c,d,e,f $tribfac0(i,n,a) &
   $tribfac0(0,n,b) & a=b & $tribfac1(i,n,c) & $tribfac1(0,n,d)
   & c=d+1 & $tribfac2(i,n,e) & $tribfac2(0,n,f) & e+1=f":
def tm12m1 "?msd_trib Ea,b,c,d,e,f $tribfac0(i,n,a) &
   $tribfac0(0,n,b) & a+1=b & $tribfac1(i,n,c) & $tribfac1(0,n,d)
   & c=d+2 & $tribfac2(i,n,e) & $tribfac2(0,n,f) & e+1=f":
def tm110 "?msd_trib Ea,b,c,d,e,f $tribfac0(i,n,a) &
   $tribfac0(0,n,b) & a+1=b & $tribfac1(i,n,c) & $tribfac1(0,n,d)
   & c=d+1 & $tribfac2(i,n,e) & $tribfac2(0,n,f) & e=f":
def t0m11 "?msd_trib Ea,b,c,d,e,f $tribfac0(i,n,a) &
   $tribfac0(0,n,b) & a=b & $tribfac1(i,n,c) & $tribfac1(0,n,d)
   & c+1=d & $tribfac2(i,n,e) & $tribfac2(0,n,f) & e=f+1":
def tm101 "?msd_trib Ea,b,c,d,e,f $tribfac0(i,n,a) &
   $tribfac0(0,n,b) & a+1=b & $tribfac1(i,n,c) & $tribfac1(0,n,d)
   & c=d & $tribfac2(i,n,e) & $tribfac2(0,n,f) & e=f+1":
def tm1m12 "?msd_trib Ea,b,c,d,e,f $tribfac0(i,n,a) &
   $tribfac0(0,n,b) & a+1=b & $tribfac1(i,n,c) & $tribfac1(0,n,d)
   & c+1=d & $tribfac2(i,n,e) & $tribfac2(0,n,f) & e=f+2":
\end{verbatim}
Interestingly enough, each of these automata has 101 states.

Next, we create a {\tt Walnut} predicate that takes two arguments, $m$ and
$n$, and returns true if $A_m = A_n$.
{\footnotesize
\begin{verbatim}
def subseteq "?msd_trib ((Ei $t000(i,m)) <=> (Ej $t000(j,n))) 
& ((Ei $t10m1(i,m)) <=> (Ej $t10m1(j,n))) & ((Ei $t1m10(i,m)) <=> (Ej $t1m10(j,n)))
& ((Ei $t01m1(i,m)) <=> (Ej $t01m1(j,n))) & ((Ei $tm12m1(i,m)) <=> (Ej $tm12m1(j,n))) 
& ((Ei $tm110(i,m)) <=> (Ej $tm110(j,n))) & ((Ei $t0m11(i,m)) <=> (Ej $t0m11(j,n))) 
& ((Ei $tm101(i,m)) <=> (Ej $tm101(j,n))) & ((Ei $tm1m12(i,m)) <=> (Ej $tm1m12(j,n)))":
\end{verbatim}
}
The resulting Tribonacci automaton has 5251 states.

Now we will iteratively determine the possible subsets of $A$
that occur as some $A_n$.
We'll do this iteratively as follows:
Starting with the subset $A_0 = \{ (0,0,0) \}$
we will find the least $n$ for which 
a different subset occurs than the ones we found previously.
Then we find the subset corresponding to this particular $n$.
This gives us all possible subsets occurring as an $A_n$ and the
smallest $n$ for which this subset occurs:
\begin{align}
A_{0} &= \{(0,0,0)\} \nonumber\\
A_{1} &= \{(-1,0,1),(-1,1,0),(0,0,0)\}\nonumber\\
A_{2} &= \{(0,-1,1),(0,0,0),(1,-1,0)\}\nonumber\\
A_{3} &= \{(-1,0,1),(-1,1,0),(0,-1,1),(0,0,0)\}\nonumber\\
A_{4} &= \{(0,0,0),(0,1,-1),(1,0,-1)\}\nonumber\\
A_{5} &= \{(-1,0,1),(-1,1,0),(0,0,0),(0,1,-1)\}\nonumber\\
A_{6} &= \{(0,-1,1),(0,0,0),(1,-1,0),(1,0,-1)\}\nonumber\\
A_{9} &= \{(-1,0,1),(0,-1,1),(0,0,0),(1,-1,0)\}\nonumber\\
A_{11} &= \{(-1,1,0),(0,0,0),(0,1,-1),(1,0,-1)\}\nonumber\\
A_{17} &= \{(0,0,0),(0,1,-1),(1,-1,0),(1,0,-1)\}\nonumber\\
A_{30} &= \{(0,-1,1),(0,0,0),(0,1,-1),(1,-1,0),(1,0,-1)\}\nonumber\\
A_{31} &= \{(-1,0,1),(-1,1,0),(0,-1,1),(0,0,0),(0,1,-1)\}\nonumber\\
A_{55} &= \{(-1,0,1),(-1,1,0),(0,0,0),(0,1,-1),(1,0,-1)\}\nonumber\\
A_{57} &= \{(-1,0,1),(0,-1,1),(0,0,0),(1,-1,0),(1,0,-1)\}\nonumber\\
A_{101} &= \{(-1,0,1),(-1,1,0),(0,-1,1),(0,0,0),(1,-1,0)\} \label{sublist}\\
A_{105} &= \{(-1,1,0),(0,0,0),(0,1,-1),(1,-1,0),(1,0,-1)\}\nonumber\\
A_{185} &= \{(-1,0,1),(-1,1,0),(-1,2,-1),(0,0,0),(0,1,-1)\}\nonumber\\
A_{340} &= \{(-1,-1,2),(-1,0,1),(0,-1,1),(0,0,0),(1,-1,0)\}\nonumber\\
A_{341} &= \{(-1,-1,2),(-1,0,1),(-1,1,0),(0,-1,1),(0,0,0)\}\nonumber\\
A_{342} &= \{(-1,0,1),(-1,1,0),(0,-1,1),(0,0,0),(0,1,-1),(1,0,-1)\}\nonumber\\
A_{355} &= \{(-1,0,1),(0,-1,1),(0,0,0),(0,1,-1),(1,-1,0),(1,0,-1)\}\nonumber\\
A_{629} &= \{(-1,0,1),(-1,1,0),(0,-1,1),(0,0,0),(1,-1,0),(1,0,-1)\}\nonumber\\
A_{653} &= \{(-1,0,1),(-1,1,0),(0,0,0),(0,1,-1),(1,-1,0),(1,0,-1)\}\nonumber\\
A_{1157} &= \{(-1,1,0),(0,-1,1),(0,0,0),(0,1,-1),(1,-1,0),(1,0,-1)\}\nonumber\\
A_{1201} &= \{(-1,0,1),(-1,1,0),(0,-1,1),(0,0,0),(0,1,-1),(1,-1,0)\}\nonumber\\
A_{3914} &= \{(-1,0,1),(-1,1,0),(0,-1,1),(0,0,0),(0,1,-1),(1,-1,0),(1,0,-1)\}\nonumber
\end{align}

For example, having computed $A_0, \ldots , A_{1201}$, to find the next
and final term we 
can use the following {\tt Walnut} code:
\begin{verbatim}
def last "?msd_trib ~($subset(n,0) | $subset(n,1) | $subset(n,2) | $subset(n,3)
| $subset(n,4) | $subset(n,5) | $subset(n,6) | $subset(n,9) | $subset(n,11)
| $subset(n,17) | $subset(n,30) | $subset(n,31) | $subset(n,55) | $subset(n,57)
| $subset(n,101) | $subset(n,105) | $subset(n,185) | $subset(n,340) 
| $subset(n,341) | $subset(n,342) | $subset(n,355) | $subset(n,629)
| $subset(n,653) | $subset(n,1157) | $subset(n,1201))":
def missing "?msd_trib $last(n) & Am (m<n) => ~$last(m)":
\end{verbatim}
Then the Tribonacci automaton computed by {\tt Walnut} and
stored in {\tt missing.txt}, depicted
below in Figure~\ref{six}, accepts
exactly one word, which is $10011000000000$, representing $3914$ in
Tribonacci representation.  So we know the next term is $A_{3914}$.
\begin{figure}[H]
\begin{center}
\vskip -.3in
\includegraphics[width=6.5in]{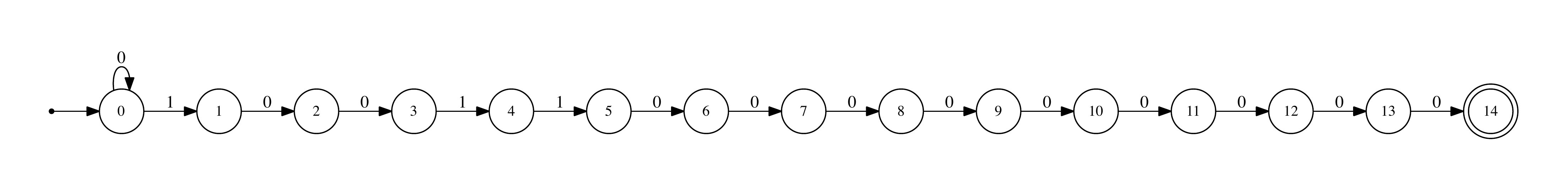}
\end{center}
\vskip -.5in
\caption{The next term is $3914$.}
\label{six}
\end{figure}
By inspecting the resulting subsets, we obtain Theorem 1.4 from
\cite{Richomme&Saari&Zamboni:2010}:   the abelian complexity
function $\rho_{\bf TR}^{\rm ab}(n)$ takes on only the values
$\{3,4,5,6,7\}$ for $n \geq 1$.

Finally, we can use the direct product construction mentioned above in 
Section~\ref{steps} to obtain
two DFAO's.   The first, corresponding to the output function $\tau_1$
computes, for each $n$, the minimal $n'$
such that $A_{n} = A_{n'}$; it is called {\tt TRAS.txt}. The second, corresponding to the
output function $\tau_2$ computes, for each $n$,
the value of $\rho_{\bf TR}^{\rm ab} (n)$; it is called {\tt TRAC.txt}.  
The transition functions
are given in Table~\ref{tab1}.   

In addition to proving Theorem 10 from \cite{Turek:2015}, we have
now also proved:
\begin{theorem}
The set-valued function $n \rightarrow \{ \psi({\bf TR}[i..i+n-1]) \suchthat
i \geq 0 \}$ is Tribonacci-automatic.
It is given by $A_n + \psi({\bf TR}[0..n-1])$.
\end{theorem}

The reason why the second automaton given in Table~\ref{tab1}
differs from that of
Turek \cite{Turek:2015} is that our automaton detects illegal
Tribonacci representations (in state 7) and produces an output
$-1$ for those representations, while Turek's automaton does not.
We have checked the output of Turek's automaton against ours
for the first 1,000,000 terms and they agree in all cases.
This serves as a double-check on our results.
Once we have the automaton, additional basic results such as
Turek's Corollary 7.2 of \cite{Turek:2013}
(asserting that
$\rho_{\bf TR}^{\rm ab} (n) = 4$ for infinitely many $n$)
can be easily obtained with {\tt Walnut}:
\begin{verbatim}
eval test4 "?msd_trib An Em (m>n) & TRAC[m]=@4":
\end{verbatim}
which evaluates to {\tt true}.    In fact, the analogous result
holds for each of the abelian complexities $3,4,5,6,7$.

Analogously, using {\tt TRAS},
we can also obtain the following new result:
\begin{theorem}
Each of the subsets in \eqref{sublist}, except $A_0$, occurs infinitely
often as the value of an $A_n$.
\end{theorem}

\begin{table}[H]
\centering
\resizebox{\columnwidth}{!}{%
\begin{tabular}{ccccc|ccccc|ccccc}
 $q$ & $\delta(q,0)$ & $\delta(q,1)$ & $\tau_1(q)$ & $\tau_2(q)$ &
 $q$ & $\delta(q,0)$ & $\delta(q,1)$ & $\tau_1(q)$ & $\tau_2(q)$ &
 $q$ & $\delta(q,0)$ & $\delta(q,1)$ & $\tau_1(q)$ & $\tau_2(q)$ \\
 \hline
  0& 0& 1& 0& 1&  26& 31& 19& 55&  5&   52&  57&  24& 629&   6\\
  1& 2& 3& 1& 3&  27& 29& 32&  5&  4&   53&  58&   7&   3&   4\\
  2& 4& 5& 2& 3&  28& 33&  5& 57&  5&   54&  59&  39& 653&   6\\
  3& 6& 7& 3& 4&  29& 16& 20&  9&  4&   55&  60&  61&   5&   4\\
  4& 8& 9& 4& 3&  30& 34& 35&  6&  4&   56&  41&  62& 185&   5\\
  5&10&11& 5& 4&  31& 36& 24&101&  5&   57&  63&  46&1157&   6\\
  6&12& 5& 6& 4&  32& 37&  7&  3&  4&   58&  64&  65&  30&   5\\
  7& 7& 7& -1&-1&  33& 38& 39&105&  5&   59&  66&  46&1201&   6\\
  8& 6&13& 3& 4&  34& 18& 40& 11&  4&   60&  49&  67&   9&   4\\
  9&14&15& 1& 3&  35& 41& 42&185&  5&   61&  68&   7& 341&   5\\
 10&16& 5& 9& 4&  36& 43& 27& 30&  5&   62&  69&   7& 341&   5\\
 11&17& 7& 3& 4&  37& 26& 44& 30&  5&   63&  70&  46& 342&   6\\
 12&18&19&11& 4&  38& 45& 46& 31&  5&   64&  31&  40&  55&   5\\
 13& 6& 3& 3& 4&  39& 47& 15& 31&  5&   65&  71&  72& 185&   5\\
 14& 4&20& 2& 3&  40& 30& 48&  3&  4&   66&  73&  27& 355&   6\\
 15&21& 7& 3& 4&  41& 49& 50&340&  5&   67&  31&  25&  55&   5\\
 16&22&23&17& 4&  42& 51&  7&341&  5&   68&  26&  74&  30&   5\\
 17&12&20& 6& 4&  43& 52& 19&342&  6&   69&  34&  65&   6&   4\\
 18&21&24& 3& 4&  44& 29& 53&  5&  4&   70&  70&  46&3914&   7\\
 19&17&15& 3& 4&  45& 54& 27&355&  6&   71&  49&  67& 340&   5\\
 20&10&25& 5& 4&  46& 47& 32& 31&  5&   72&  58&   7& 341&   5\\
 21&26&27&30& 5&  47& 33& 20& 57&  5&   73&  70&  39&3914&   7\\
 22&28&19&31& 5&  48& 36&  7&101&  5&   74&  60&  75&   5&   4\\
 23&29&15& 5& 4&  49& 22& 55& 17&  4&   75&  76&   7& 341&   5\\
 24&28&11&31& 5&  50& 31& 11& 55&  5&   76&  26&  77&  30&   5\\
 25&30& 7& 3& 4&  51& 34& 56&  6&  4&   77&  60&  72&   5&   4
\end{tabular}
}
\caption{The two DFAO's computing $A_n$ and
$\rho_{\bf TR}^{\rm ab} (n)$.}
\label{tab1}
\end{table}

All the {\tt Walnut} code referred to in this paper is available
from\\
\centerline{\url{https://cs.uwaterloo.ca/~shallit/papers.html}}

\section{Conclusions}

We have shown that in some cases the abelian complexity function of an
automatic sequence can be ``automatically'' computed.   This continues
in the spirit of other papers (e.g., \cite{Charlier&Rampersad&Shallit:2011,Mousavi&Shallit:2015})
that try to automate results in
combinatorics on words that formerly needed extensive case analysis to prove.

\section{Acknowledgments}

We are pleased to thank 
Ond\v{r}ej Turek for his helpful comments.

\end{document}